%% file: main_arxiv.tex
\documentclass[a4paper,11pt]{article}

\usepackage[T1]{fontenc}
\usepackage[utf8]{inputenc}
\usepackage[UKenglish]{babel}
\usepackage{amsmath,amssymb,amsthm}
\usepackage[ruled,vlined,linesnumbered]{algorithm2e}
\SetArgSty{textnormal}
\usepackage{tikz}
\usetikzlibrary{positioning,shadows}
\usepackage{hyperref}
\usepackage[nameinlink,noabbrev]{cleveref}
\usepackage{xspace}
\usepackage{xcolor}
\usepackage[margin=1in]{geometry}
\usepackage{microtype}

\newtheorem{theorem}{Theorem}[section]
\newtheorem{lemma}[theorem]{Lemma}

\input{macros.tex}

\title{Polylogarithmic Approximation for Covering and Connecting Multi-Interface Networks}

\author{Micha\l{} Szyfelbein\thanks{Gda\'nsk University of Technology, Poland. \texttt{michal.szyfelbein@pg.edu.pl}. \href{https://orcid.org/0009-0009-9894-9671}{ORCID: 0009-0009-9894-9671}}
\and
Camille Richer\thanks{Universit\'e Paris-Dauphine, PSL Research University, CNRS, UMR 7243, LAMSADE. \texttt{camille.richer@dauphine.eu}. \href{https://orcid.org/0009-0000-3636-6571}{ORCID: 0009-0000-3636-6571}}}

\date{}

\begin{document}

\maketitle

\input{sections/abstract.tex}

\input{sections/introduction.tex}
\input{sections/preliminaries.tex}
\input{sections/lp-based-approximation-algorithms.tex}

\input{sections/conclusion.tex}
\bibliography{lipics-v2021-sample-article}

\appendix
\input{sections/appendix.tex}

\end{document}

%% file: macros.tex
\tikzstyle{mybox} = [draw=black, fill=white, drop shadow, very thick, rectangle, rounded corners, inner sep=10pt, inner ysep=12pt]
\tikzstyle{fancytitle} =[fill=white, text=black]
\newcommand{\pbdef}[3]{
\noindent
\begin{tikzpicture}
\node [mybox] (box){%
    \begin{minipage}{0.94\textwidth}
        \textit{Input: }{#2}
        \\
        \textit{Task: }{#3}
    \end{minipage}
};
\node[fancytitle, right=10pt] at (box.north west) {#1};
\end{tikzpicture}
}

\newcommand{\coverage}{\textsc{Coverage}\xspace}
\newcommand{\connectivity}{\textsc{Connectivity}\xspace}
\newcommand{\setcover}{\textsc{Set Cover}\xspace}
\newcommand{\mst}{\textsc{Minimum Spanning Tree}\xspace}

\newcommand{\br}[1]{\mathopen{}\left( #1 \right)}
\newcommand{\sbr}[1]{\mathopen{}\left[ #1 \right]}
\newcommand{\brc}[1]{\mathopen{}\left\{ #1 \right\}}
\newcommand{\spr}[1]{\mathopen{}\left| #1 \right|}

\newcommand{\cl}[1]{\mathopen{}\left\lceil #1 \right\rceil}

\newcommand{\e}[1]{\exp\left\{ #1 \right\}}

\newcommand{\NPhard}{\textnormal{NP-Hard}\xspace}

\newcommand{\OPT}{\textnormal{OPT}}
\newcommand{\COST}{\textnormal{cost}}

\newcommand{\cov}{\operatorname{cov}}
\newcommand{\E}[1]{\mathbb{E}\!\left[#1\right]}

%% file: sections/abstract.tex
\begin{abstract}
We study problems related to connecting multi-interface networks of wireless devices. These problems can be modeled using graphs, where vertices represent the devices and edges represent potential communication  links. Each vertex can activate multiple interfaces, and a connection between two vertices is established if they share at least one common active interface.
However, activating an interface induces a cost that depends both on the type of the interface and on the vertex that activates it. We consider two problems arising in multi-interface networks: \coverage and \connectivity. In the \coverage problem, every connection defined in the network must be established, while in the \connectivity problem, it is only required that the established connections form a subgraph spanning the network. The solution should also minimize either the maximum cost incurred by a node or the total cost incurred by all vertices. In this work we are interested in approximating the former of the two cost criterions.

We model both problems using Integer Linear Programming (ILP) and we design approximation algorithms based on a randomized rounding of the solution of the linear programming relaxation. For the \coverage problem, this yields an $O(\log n)$-approximation algorithm, where $n$ is the number of vertices. This result is tight, since the problem generalizes \setcover. This improves upon the $O(b\cdot\log n)$-approximation algorithm, where $b$ is a certain graph parameter which can be as large as $\Omega(n)$ [Algorithmica '12]. The same relaxation can also be used to get a $k$-approximation algorithm, where $k$ is the number of different interfaces. This generalizes a similar result for the homogeneous cost case, where the cost of an interface is the same for all vertices.
The main result of our work is an $O(\log^2 n)$-approximation algorithm for the \connectivity, which is the first non-trivial approximation for this problem. The algorithm is based on a similar LP relaxation with additional cut constraints to ensure connectivity. The rounding procedure resembles the one for the \coverage but requires a more careful analysis to ensure that the connectivity constraints are satisfied.
\end{abstract}

%% file: sections/introduction.tex
\section{Introduction}
Designing Internet of Things (IoT) networks often demands establishing connections between large number of devices, which can be achieved by activating communication interfaces at each of them. However, the devices in the network may vary according to the types of interfaces they have at their availability and the amount of resources a given interface consumes. For example some device may be able to utilize Bluetooth and Wi-Fi connections, while the others may provide Zigbee and Z-Wave interfaces. Since IoT devices are typically resource-constrained, it may be desirable to try to establish every relevant connection, while minimizing the amount of resources required to do so. 

\begin{figure}[t]
\centering
\input{figures/combined-illustration.tex}
\caption{From left to right: a common input graph, a covering assignment, and a connecting assignment for the same instance. An edge is colored according to an interface that covers it.}
\label{fig:example_problems}
\end{figure}
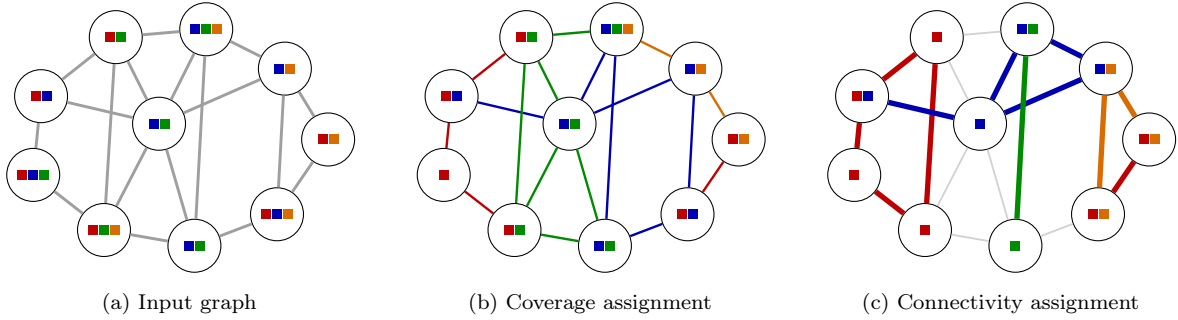
We model the above problem using a graph. Every device in the network is represented by a vertex with a set of labels encoding its available interfaces. Two devices can communicate if they are connected by an edge, representing that they are within communication range. A connection is established between two adjacent devices whenever they share a common active interface. 
We want to find an assignment of active interfaces to vertices so that we ensure either the \coverage or \connectivity property. 

In the \coverage problem, we wish to establish every possible direct connection in the network which means that both endpoints of every edge must share a common active interface (see the middle panel in \Cref{fig:example_problems} for an example) while in the \connectivity problem, we require all vertices to be (not necessarily directly) connected in the subgraph induced by active interfaces (see the right panel in \Cref{fig:example_problems} for an example). In both problems, we assume that activating an interface induces a cost that depends on the type of interface and on the vertex that activates it. The goal is to minimize the max cost, i.~e., the maximum cost of  interfaces activated at a given vertex across all vertices. 

\subsection{Related Work}

In the following, let $k$ denote the number interfaces types across the network, $n$ the number of vertices of the graph, $m$ the number of edges, $\Delta$ the maximum degree, $c_{\max}$ the maximum cost and $c_{\min}$ the minimum cost of an interface. In the \emph{homogeneous} setting, the costs only depend on the interface while in the \emph{heterogeneous} case they also depend on the vertex. 

\subparagraph*{Coverage Problem.}
The algorithmic study of problems in multi-interface networks started with the problem of covering a multi-interface network, introduced in Caporuscio et al. \cite{EnergeticPerformanceOfServiceOrientedMultiRadioNetworks}. 
This problem has been mainly studied in the homogeneous setting, both for min-sum \cite{CostMinimizationInWirelessNetworks} and min-max \cite{MinimizeTheMaximumDutyInMultiInterfaceNetworks} objective functions.
The min-sum variant is APX-hard for any constant $k \geq 3$ and unit costs, and it is hard to approximate within an $o\br{\ln k}$ factor even with unit costs, but it has a $\min \brc{k-1, \sqrt{n} \br{1+\ln n}}$-approximation \cite{CostMinimizationInWirelessNetworks}.
The min-max variant is \NPhard for any constant $\Delta \geq 5$, $k \geq 16$ and for unit costs, and it is not approximable within a $\eta \ln \Delta$ factor for some constant $\eta$ even on trees and with unit costs. On the positive side, using an approximation of \setcover, it admits a $\br{\ln \Delta +1 +b \cdot \min \brc{c_{\max}, \ln \Delta +1}}$-approximation, where $b$ is a parameter that depends on the structural properties of the input graph and is bounded by many other structutal parameters such as treewidth, arboricity or maximum degree. It also has simple $(1+\br{k-2} \frac{c_{\max}}{2 c_{\min}})$- and $(\frac{\Delta}{2} \frac{c_{\max}}{c_{\min}})$-approximation \cite{MinimizeTheMaximumDutyInMultiInterfaceNetworks}. 

\subparagraph*{Connectivity Problem.} The problem of connecting a multi-interface network was introduced by Kosowski et al. \cite{ExploitingMultiInterfaceNetworksConnectivityAndCheapestPaths}. It was first studied in the homogeneous setting, both for the min-sum \cite{ExploitingMultiInterfaceNetworksConnectivityAndCheapestPaths} and the min-max objective functions \cite{MinimizeTheMaximumDutyInMultiInterfaceNetworks}.
The min-sum variant with homogeneous unit costs is APX-hard for $k \geq 2$ and $\Delta \geq 4$ already, but it has a $2$-approximation \cite{ExploitingMultiInterfaceNetworksConnectivityAndCheapestPaths}. This result was improved by Athanassopoulos et al. \cite{EnergyEfficientCommunicationInMultiInterfaceWirelessNetworks} to a randomized $(\frac{3}{2} + \epsilon)$-approximation for any constant $\epsilon$, using a reduction to a \mst problem on hypergraphs. Moreover, this approximation algorithm applies to the more general problem where for each interface, the input specifies what set of edges it can cover.

The min-max variant with homogeneous unit costs is \NPhard for any constant $\Delta \geq 3$ and $k \geq 10$. Similarly to the coverage, it is not approximable within a $\eta \ln n$ factor even for unit costs, but has simple $(1+(k-2) \frac{c_{\max}}{2 c_{\min}})$- and $\frac{\Delta}{2} \frac{c_{\max}}{c_{\min}}$-approximations \cite{MinimizeTheMaximumDutyInMultiInterfaceNetworks}.

Athanassopoulos et al. \cite{EnergyEfficientCommunicationInMultiInterfaceWirelessNetworks} expand the study of the min-sum variant to group connectivity, where only certain vertices, called \textit{terminals} need to be spanned and there may be several terminal groups. They provide a randomized $4$-approximation for this case. They also initiate the study of this problem in the heterogeneous setting. They show that for heterogeneous costs, the min-sum variant is hard to approximate within $o(\ln n)$ even when all vertices are terminals, but the min-sum group connectivity variant admits a $O(\ln n)$-approximation.

\subparagraph*{Other Problems.} There has been some  study of the coverage problem from the perspective of parameterized complexity, possibly with an additional constraint on the maximum number of active interfaces at a vertex \cite{MinMaxCoverageInMultiInterfaceNetworksPathwidth,OnComputationalComplexityOfCoveringMultiInterfaceNetworks,SurveyOfCoverageProblemsinMultiInterfaceWirelessNetworks}, yielding efficient algorithms for some graph classes. For example, a variant with a constraint on the cost of a solution is considered, where the goal is to maximize some profit function \cite{SurveyOfCoverageProblemsinMultiInterfaceWirelessNetworks}. Some papers also study other problems in multi-interface networks, like cheapest path \cite{ExploitingMultiInterfaceNetworksConnectivityAndCheapestPaths}, matchings \cite{MaximumMatchingInMultiInterfaceNetworks}, or flow problems \cite{FlowProblemsInMultiInterfaceNetworks}.

\subsection{Our contribution and techniques}

We study the approximability of the \coverage and \connectivity problems minimizing the max-cost with heterogeneous cost functions. We give a \setcover-style LP relaxation for the \coverage problem and show that a randomized rounding procedure yields an $O\br{\log n}$-approximation with high probability. This is the first logarithmic approximation for this problem with heterogeneous costs, and it matches the previously known lower bound of $\Omega\br{\log n}$ for stars with unit costs \cite{MinimizeTheMaximumDutyInMultiInterfaceNetworks}. It also improves the previously best known $O\br{\ln \Delta + b \cdot \min\brc{\ln \Delta, c_{\max}}}$-approximation ratio \cite{MinimizeTheMaximumDutyInMultiInterfaceNetworks} (for some graph input parameter $b$), which was limited to the homogeneous case. 
A simpler deterministic rounding of the same LP gives a $k$-approximation, generalizing the $1+(k-2) \frac{c_{\max}}{2 c_{\min}}$-approximation for the min-max variant with homogeneous costs \cite{MinimizeTheMaximumDutyInMultiInterfaceNetworks}.

For the \connectivity problem, we extend the LP by adding flow-based cut constraints encoding global connectivity. Our rounding procedure is similar to the one for the \coverage problem; however, this time the process is iterative and the algorithm proceeds until all vertices are connected. Each iteration is randomized and we show that in expectation, in each such round the algorithm covers a significant fraction of edges of a graph that spans all vertices. This yields an $O\br{\log^2 n}$-approximation algorithm, which is the first non-trivial approximation algorithm for this problem.

\subsection{Paper Organization}

The rest of the paper is structured as follows:
In \Cref{sec:preliminaries}, we give notations, define the problems formally, and recall the elements of probability theory useful for the analysis. We also describe the preprocessing required for our randomized algorithms for the \coverage and \connectivity problems to work.
In \Cref{sec:coverage}, we present the $O\br{\log n}$-approximation and the $k$-approximation algorithms for the \coverage problem.
In \Cref{sec:connectivity}, we present the $O\br{\log^2 n}$-approximation algorithm for the \connectivity problem.
We conclude with a short discussion on open questions and future work.

%% file: figures/combined-illustration.tex
\begin{minipage}[t]{0.31\linewidth}
\centering
\begin{tikzpicture}[scale=0.52]
    \colorlet{iface1}{red!75!black}
    \colorlet{iface2}{blue!70!black}
    \colorlet{iface3}{green!55!black}
    \colorlet{iface4}{orange!85!black}

    \node[circle, draw, minimum size=7mm, inner sep=0pt] (v1) at (0.0,1.2) {};
    \node[circle, draw, minimum size=7mm, inner sep=0pt] (v2) at (1.9,2.7) {};
    \node[circle, draw, minimum size=7mm, inner sep=0pt] (v3) at (4.2,2.9) {};
    \node[circle, draw, minimum size=7mm, inner sep=0pt] (v4) at (6.2,1.9) {};
    \node[circle, draw, minimum size=7mm, inner sep=0pt] (v5) at (7.3,0.1) {};
    \node[circle, draw, minimum size=7mm, inner sep=0pt] (v6) at (6.0,-1.8) {};
    \node[circle, draw, minimum size=7mm, inner sep=0pt] (v7) at (3.9,-2.6) {};
    \node[circle, draw, minimum size=7mm, inner sep=0pt] (v8) at (1.6,-2.2) {};
    \node[circle, draw, minimum size=7mm, inner sep=0pt] (v9) at (-0.2,-0.8) {};
    \node[circle, draw, minimum size=7mm, inner sep=0pt] (v10) at (3.0,0.5) {};

    \draw[gray!75, line width=1.1pt] (v1) -- (v2) -- (v3) -- (v4) -- (v5) -- (v6) -- (v7) -- (v8) -- (v9) -- (v1);
    \draw[gray!75, line width=1.1pt] (v1) -- (v10);
    \draw[gray!75, line width=1.1pt] (v2) -- (v10);
    \draw[gray!75, line width=1.1pt] (v3) -- (v10);
    \draw[gray!75, line width=1.1pt] (v4) -- (v10);
    \draw[gray!75, line width=1.1pt] (v7) -- (v10);
    \draw[gray!75, line width=1.1pt] (v8) -- (v10);
    \draw[gray!75, line width=1.1pt] (v2) -- (v8);
    \draw[gray!75, line width=1.1pt] (v3) -- (v7);
    \draw[gray!75, line width=1.1pt] (v4) -- (v6);

    \begin{scope}[shift={(v1.center)}]
        \fill[iface1] (-0.26,-0.12) rectangle (-0.02,0.12);
        \fill[iface2] (0.02,-0.12) rectangle (0.26,0.12);
    \end{scope}
    \begin{scope}[shift={(v2.center)}]
        \fill[iface1] (-0.26,-0.12) rectangle (-0.02,0.12);
        \fill[iface3] (0.02,-0.12) rectangle (0.26,0.12);
    \end{scope}
    \begin{scope}[shift={(v3.center)}]
        \fill[iface2] (-0.40,-0.12) rectangle (-0.16,0.12);
        \fill[iface3] (-0.12,-0.12) rectangle (0.12,0.12);
        \fill[iface4] (0.16,-0.12) rectangle (0.40,0.12);
    \end{scope}
    \begin{scope}[shift={(v4.center)}]
        \fill[iface2] (-0.26,-0.12) rectangle (-0.02,0.12);
        \fill[iface4] (0.02,-0.12) rectangle (0.26,0.12);
    \end{scope}
    \begin{scope}[shift={(v5.center)}]
        \fill[iface1] (-0.26,-0.12) rectangle (-0.02,0.12);
        \fill[iface4] (0.02,-0.12) rectangle (0.26,0.12);
    \end{scope}
    \begin{scope}[shift={(v6.center)}]
        \fill[iface1] (-0.40,-0.12) rectangle (-0.16,0.12);
        \fill[iface2] (-0.12,-0.12) rectangle (0.12,0.12);
        \fill[iface4] (0.16,-0.12) rectangle (0.40,0.12);
    \end{scope}
    \begin{scope}[shift={(v7.center)}]
        \fill[iface2] (-0.26,-0.12) rectangle (-0.02,0.12);
        \fill[iface3] (0.02,-0.12) rectangle (0.26,0.12);
    \end{scope}
    \begin{scope}[shift={(v8.center)}]
        \fill[iface1] (-0.40,-0.12) rectangle (-0.16,0.12);
        \fill[iface3] (-0.12,-0.12) rectangle (0.12,0.12);
        \fill[iface4] (0.16,-0.12) rectangle (0.40,0.12);
    \end{scope}
    \begin{scope}[shift={(v9.center)}]
        \fill[iface1] (-0.40,-0.12) rectangle (-0.16,0.12);
        \fill[iface2] (-0.12,-0.12) rectangle (0.12,0.12);
        \fill[iface3] (0.16,-0.12) rectangle (0.40,0.12);
    \end{scope}
    \begin{scope}[shift={(v10.center)}]
        \fill[iface2] (-0.26,-0.12) rectangle (-0.02,0.12);
        \fill[iface3] (0.02,-0.12) rectangle (0.26,0.12);
    \end{scope}
\end{tikzpicture}

{\scriptsize (a) Input graph}
\end{minipage}
\hfill
\begin{minipage}[t]{0.31\linewidth}
\centering
\begin{tikzpicture}[scale=0.52]
    \colorlet{iface1}{red!75!black}
    \colorlet{iface2}{blue!70!black}
    \colorlet{iface3}{green!55!black}
    \colorlet{iface4}{orange!85!black}

    \node[circle, draw, minimum size=7mm, inner sep=0pt] (v1) at (0.0,1.2) {};
    \node[circle, draw, minimum size=7mm, inner sep=0pt] (v2) at (1.9,2.7) {};
    \node[circle, draw, minimum size=7mm, inner sep=0pt] (v3) at (4.2,2.9) {};
    \node[circle, draw, minimum size=7mm, inner sep=0pt] (v4) at (6.2,1.9) {};
    \node[circle, draw, minimum size=7mm, inner sep=0pt] (v5) at (7.3,0.1) {};
    \node[circle, draw, minimum size=7mm, inner sep=0pt] (v6) at (6.0,-1.8) {};
    \node[circle, draw, minimum size=7mm, inner sep=0pt] (v7) at (3.9,-2.6) {};
    \node[circle, draw, minimum size=7mm, inner sep=0pt] (v8) at (1.6,-2.2) {};
    \node[circle, draw, minimum size=7mm, inner sep=0pt] (v9) at (-0.2,-0.8) {};
    \node[circle, draw, minimum size=7mm, inner sep=0pt] (v10) at (3.0,0.5) {};

    \draw[iface1, line width=0.9pt] (v1) -- (v2);
    \draw[iface3, line width=0.9pt] (v2) -- (v3);
    \draw[iface4, line width=0.9pt] (v3) -- (v4);
    \draw[iface4, line width=0.9pt] (v4) -- (v5);
    \draw[iface1, line width=0.9pt] (v5) -- (v6);
    \draw[iface2, line width=0.9pt] (v6) -- (v7);
    \draw[iface3, line width=0.9pt] (v7) -- (v8);
    \draw[iface1, line width=0.9pt] (v8) -- (v9);
    \draw[iface1, line width=0.9pt] (v9) -- (v1);
    \draw[iface2, line width=0.9pt] (v1) -- (v10);
    \draw[iface3, line width=0.9pt] (v2) -- (v10);
    \draw[iface2, line width=0.9pt] (v3) -- (v10);
    \draw[iface2, line width=0.9pt] (v4) -- (v10);
    \draw[iface3, line width=0.9pt] (v7) -- (v10);
    \draw[iface3, line width=0.9pt] (v8) -- (v10);
    \draw[iface3, line width=0.9pt] (v2) -- (v8);
    \draw[iface2, line width=0.9pt] (v3) -- (v7);
    \draw[iface2, line width=0.9pt] (v4) -- (v6);

    \begin{scope}[shift={(v1.center)}]
        \fill[iface1] (-0.26,-0.12) rectangle (-0.02,0.12);
        \fill[iface2] (0.02,-0.12) rectangle (0.26,0.12);
    \end{scope}
    \begin{scope}[shift={(v2.center)}]
        \fill[iface1] (-0.26,-0.12) rectangle (-0.02,0.12);
        \fill[iface3] (0.02,-0.12) rectangle (0.26,0.12);
    \end{scope}
    \begin{scope}[shift={(v3.center)}]
        \fill[iface2] (-0.40,-0.12) rectangle (-0.16,0.12);
        \fill[iface3] (-0.12,-0.12) rectangle (0.12,0.12);
        \fill[iface4] (0.16,-0.12) rectangle (0.40,0.12);
    \end{scope}
    \begin{scope}[shift={(v4.center)}]
        \fill[iface2] (-0.26,-0.12) rectangle (-0.02,0.12);
        \fill[iface4] (0.02,-0.12) rectangle (0.26,0.12);
    \end{scope}
    \begin{scope}[shift={(v5.center)}]
        \fill[iface1] (-0.26,-0.12) rectangle (-0.02,0.12);
        \fill[iface4] (0.02,-0.12) rectangle (0.26,0.12);
    \end{scope}
    \begin{scope}[shift={(v6.center)}]
        \fill[iface1] (-0.26,-0.12) rectangle (-0.02,0.12);
        \fill[iface2] (0.02,-0.12) rectangle (0.26,0.12);
    \end{scope}
    \begin{scope}[shift={(v7.center)}]
        \fill[iface2] (-0.26,-0.12) rectangle (-0.02,0.12);
        \fill[iface3] (0.02,-0.12) rectangle (0.26,0.12);
    \end{scope}
    \begin{scope}[shift={(v8.center)}]
        \fill[iface1] (-0.26,-0.12) rectangle (-0.02,0.12);
        \fill[iface3] (0.02,-0.12) rectangle (0.26,0.12);
    \end{scope}
    \begin{scope}[shift={(v9.center)}]
        \fill[iface1] (-0.12,-0.12) rectangle (0.12,0.12);
    \end{scope}
    \begin{scope}[shift={(v10.center)}]
        \fill[iface2] (-0.26,-0.12) rectangle (-0.02,0.12);
        \fill[iface3] (0.02,-0.12) rectangle (0.26,0.12);
    \end{scope}
\end{tikzpicture}

{\scriptsize (b) Coverage assignment}
\end{minipage}
\hfill
\begin{minipage}[t]{0.31\linewidth}
\centering
\begin{tikzpicture}[scale=0.52]
    \colorlet{iface1}{red!75!black}
    \colorlet{iface2}{blue!70!black}
    \colorlet{iface3}{green!55!black}
    \colorlet{iface4}{orange!85!black}

    \node[circle, draw, minimum size=7mm, inner sep=0pt] (v2)  at (1.9, 2.7) {};
    \node[circle, draw, minimum size=7mm, inner sep=0pt] (v6)  at (6.0,-1.8) {};
    \node[circle, draw, minimum size=7mm, inner sep=0pt] (v9)  at (-0.2,-0.8) {};
    \node[circle, draw, minimum size=7mm, inner sep=0pt] (v3) at (4.2, 2.9) {};
    \node[circle, draw, minimum size=7mm, inner sep=0pt] (v7) at (3.9,-2.6) {};
    \node[circle, draw, minimum size=7mm, inner sep=0pt] (v1) at (0.0, 1.2) {};
    \node[circle, draw, minimum size=7mm, inner sep=0pt] (v5) at (7.3, 0.1) {};
    \node[circle, draw, minimum size=7mm, inner sep=0pt] (v8) at (1.6,-2.2) {};
    \node[circle, draw, minimum size=7mm, inner sep=0pt] (v4)  at (6.2, 1.9) {};
    \node[circle, draw, minimum size=7mm, inner sep=0pt] (v10) at (3.0, 0.5) {};

    \draw[gray!35, line width=0.7pt] (v2) -- (v3);
    \draw[gray!35, line width=0.7pt] (v2) -- (v10);
    \draw[gray!35, line width=0.7pt] (v7) -- (v10);
    \draw[gray!35, line width=0.7pt] (v7) -- (v8);
    \draw[gray!35, line width=0.7pt] (v8) -- (v10);
    \draw[gray!35, line width=0.7pt] (v6) -- (v7);

    \draw[iface1, line width=2.0pt] (v1)  -- (v2);
    \draw[iface2, line width=2.0pt] (v3)  -- (v4);
    \draw[iface2, line width=2.0pt] (v3)  -- (v10);
    \draw[iface2, line width=2.0pt] (v10) -- (v4);
    \draw[iface4, line width=2.0pt] (v4)  -- (v5);
    \draw[iface4, line width=2.0pt] (v4)  -- (v6);
    \draw[iface1, line width=2.0pt] (v5)  -- (v6);
    \draw[iface3, line width=2.0pt] (v3)  -- (v7);
    \draw[iface1, line width=2.0pt] (v2)  -- (v8);
    \draw[iface1, line width=2.0pt] (v8)  -- (v9);
    \draw[iface1, line width=2.0pt] (v9)  -- (v1);
    \draw[iface2, line width=2.0pt] (v1)  -- (v10);

    \begin{scope}[shift={(v1.center)}]
        \fill[iface1] (-0.26,-0.12) rectangle (-0.02,0.12);
        \fill[iface2] (0.02,-0.12) rectangle (0.26,0.12);
    \end{scope}
    \begin{scope}[shift={(v2.center)}]
        \fill[iface1] (-0.12,-0.12) rectangle (0.12,0.12);
    \end{scope}
    \begin{scope}[shift={(v3.center)}]
        \fill[iface2] (-0.26,-0.12) rectangle (-0.02,0.12);
        \fill[iface3] (0.02,-0.12) rectangle (0.26,0.12);
    \end{scope}
    \begin{scope}[shift={(v4.center)}]
        \fill[iface2] (-0.26,-0.12) rectangle (-0.02,0.12);
        \fill[iface4] (0.02,-0.12) rectangle (0.26,0.12);
    \end{scope}
    \begin{scope}[shift={(v5.center)}]
        \fill[iface1] (-0.26,-0.12) rectangle (-0.02,0.12);
        \fill[iface4] (0.02,-0.12) rectangle (0.26,0.12);
    \end{scope}
    \begin{scope}[shift={(v6.center)}]
        \fill[iface1] (-0.26,-0.12) rectangle (-0.02,0.12);
        \fill[iface4] (0.02,-0.12) rectangle (0.26,0.12);
    \end{scope}
    \begin{scope}[shift={(v7.center)}]
        \fill[iface3] (-0.12,-0.12) rectangle (0.12,0.12);
    \end{scope}
    \begin{scope}[shift={(v8.center)}]
        \fill[iface1] (-0.12,-0.12) rectangle (0.12,0.12);
    \end{scope}
    \begin{scope}[shift={(v9.center)}]
        \fill[iface1] (-0.12,-0.12) rectangle (0.12,0.12);
    \end{scope}
    \begin{scope}[shift={(v10.center)}]
        \fill[iface2] (-0.12,-0.12) rectangle (0.12,0.12);
    \end{scope}
\end{tikzpicture}

{\scriptsize (c) Connectivity assignment}
\end{minipage}

%% file: sections/preliminaries.tex
\section{Preliminaries} \label{sec:preliminaries}

Let $G=\br{V, E}$ be a connected simple graph. Let $\sbr{k}$ be the set of interfaces across the network and let $\lambda\colon V \to 2^{\sbr{k}}$ be a function that assigns to each vertex $v\in V$ its set of available interfaces, such that for each edge $uv \in E$, $\lambda\br{u} \cap \lambda\br{v} \neq \emptyset$.
An assignment of active interfaces is a function $\lambda_A\colon V \to 2^{\sbr{k}}$ such that $\lambda_A\br{v} \subseteq \lambda\br{v}$ for every $v\in V$. We say that an edge $uv \in E$ is \emph{covered} if $\lambda_A\br{u} \cap \lambda_A\br{v} \neq \emptyset$. We denote by $\cov\br{\lambda_A} = \{uv \in E : \lambda_A\br{u} \cap \lambda_A\br{v} \neq \emptyset\}$ the set of edges covered by $\lambda_A$, and by $G_A=(V,\cov\br{\lambda_A})$. 

In the \coverage problem, we wish to find an assignment of active interfaces $\lambda_A$ such that $G_A=G$. We call such an assignment a \emph{covering assignment}. In the \connectivity problem, we wish to find an assignment of active interfaces $\lambda_A$ such that $G_A$ is connected. We call such an assignment a \emph{connecting assignment}, i.e., an assignment that induces a connected spanning subgraph of $G$.

We are also given a cost function $c\colon \sbr{k}\times V \to \mathbb{N}_{\geq 0}$, which denotes the cost of activating the interface $i\in\sbr{k}$ at the vertex $v\in V$. 
For an assignment $\lambda_A$, define the \emph{cost of $v$} as:
$$\COST\br{v, \lambda_A} = \sum_{i\in\lambda_A\br{v}}c\br{i, v}.$$
The \emph{max-cost} of an assignment $\lambda_A$ is defined as:
$$\COST \br{\lambda_A}=\max_{v\in V} \COST\br{v, \lambda_A}.$$
The goal is to find an assignment $\lambda_A$ that minimizes $\COST \br{\lambda_A}$ while ensuring either \coverage or \connectivity. 
The formal problem statements are as follows.

\pbdef{\coverage}{A connected simple graph $G=(V,E)$, available interfaces $\lambda\colon V \to 2^{\sbr{k}}$, a cost function $c\colon \sbr{k}\times V \to \mathbb{N}_{\geq 0}$.}{Find a covering assignment $\lambda_A$ minimizing $\COST\br{\lambda_A}$.}

\pbdef{\connectivity}{A connected simple graph $G=(V,E)$, available interfaces $\lambda\colon V \to 2^{\sbr{k}}$, a cost function $c\colon \sbr{k}\times V \to \mathbb{N}_{\geq 0}$.}{Find a connecting assignment $\lambda_A$ minimizing $\COST\br{\lambda_A}$.}

We denote by $\OPT$ the cost of an optimal solution to the problem at hand. We will refer to the min-max covering assignment problem as \coverage and the min-max connecting assignment problem as \connectivity. Note that there also exists a min-sum variant of these problems, where the goal is to minimize $\sum_{v\in V} \COST\br{v, \lambda_A}$ instead of $\COST\br{\lambda_A}$.
\subsection{Probabilistic tools}

Since our algorithms are randomized, we will need the following probability facts. 
\begin{theorem}[Markov's Inequality]
Let $X$ be a non-negative random variable. Then, for $a > 0$:
\[
    \Pr\sbr{X \geq a} \leq \frac{\mathbb{E}\sbr{X}}{a}.
\]
\end{theorem}


\begin{theorem}[(Weighted) Chernoff Bound]
    Let $X_1, X_2, \dots, X_n$ be independent random variables taking values in $\sbr{0,1}$, let $w_1, w_2, \dots, w_n$ be non-negative weights, such that for every $i\in \sbr{n}$ we have that $w_i\in \sbr{0,1}$. Let $X=\sum_{i=1}^{n} w_i\cdot X_i$ and let $\mu = \E{X}$. Then for every $\delta_1 > 0$ and $\delta_2 \in (0,1)$ we have that:
    \[
        \Pr\br{X \geq (1+\delta_1)\cdot \mu} \leq \br{\frac{e^{\delta_1}}{(1+\delta_1)^{1+\delta_1}}}^\mu
        \qquad
        \Pr\br{X \leq (1-\delta_2)\cdot \mu} \leq \br{\frac{e^{-\delta_2}}{(1-\delta_2)^{1-\delta_2}}}^\mu.
    \]
\end{theorem}

The proof of the above fact follows the proof of the standard Chernoff bound and we include it in \Cref{sec:appendix} for completeness.

\subsection{Preprocessing}
Before applying our algorithms (for both \coverage and \connectivity), we preprocess the instance in the following way: First, by a polynomial number of guesses and rescaling, we reduce the instances so that all of the costs are rational values in $\sbr{0,1}$ and we have $\OPT\geq 1/2$. Second, we partition vertices into two classes. A vertex $v$ is \emph{cheap} if (in the scaled instance) $\sum_{i\in\lambda\br{v}} c\br{i,v}\leq 1$; for such vertices we may assume all available interfaces are activated. Vertices that are not cheap are \emph{expensive}; for them we may assume $\COST\br{v,\lambda_A}\geq 1$ in the analyzed solutions, while incurring only a factor $2$ loss in the approximation ratio. The full preprocessing procedure and its complete analysis are deferred to \Cref{sec:appendix-preprocessing}.

%% file: sections/lp-based-approximation-algorithms.tex
\section{Warmup: Approximation algorithms for the coverage problem} \label{sec:coverage}

In this section, we consider the \coverage problem. According to the preprocessing, we assume that all costs are in $\sbr{0,1}$ and $\OPT \geq 1/2$. 
The idea behind the ILP formulation and the algorithms is inspired by the rounding procedure known for \textsc{Set Cover} \cite{vazirani2001approximation} and related problems \cite{Angelidakis2018ShortestPQ}.

\subsection{ILP formulations of the \coverage problem}

We start with the ILP formulation of the \coverage problem. Let $x_{i, v}$ be a binary variable that is $1$ iff the interface $i\in\lambda\br{v}$ is assigned to vertex $v\in V$. If a given vertex $v$ is cheap, then without large cost increase we can enforce $x_{i, v}=1$, for every $i\in\lambda\br{v}$. Otherwise, if a vertex $v$ is expensive, we can assume that in an optimal solution we have $\sum_{i\in \lambda\br{v}} x_{i, v}\cdot c\br{i, v} \geq 1$, which corresponds to the fact that $\COST\br{\lambda_A, v}\geq 1$. Thus, we can write the following ILP formulation of the \coverage problem:
\begin{align}
    \text{min} \quad & M \nonumber\\
    \text{s.t.} \quad & \sum_{i\in\lambda\br{v}}c\br{i, v}\cdot x_{i, v} \leq M \quad \forall v\in V \label{eqt:covering_cost}\\
    & \sum_{i\in\lambda\br{u}\cap\lambda\br{v}}\min\brc{x_{i, u}, x_{i, v}} \geq 1 \quad \forall uv\in E \label{eqt:coverage}\\
    &\begin{cases}
         x_{i, v} = 1 \quad \forall i\in\lambda\br{v} & \text{if } \sum_{i\in \lambda\br{v}} c\br{i, v} < 1\\
        \sum_{i\in \lambda\br{v}} c\br{i, v}\cdot x_{i, v} \geq 1 & \text{otherwise}
    \end{cases}\quad \forall v\in V \label{eqt:costs_coverage}\\
    & x_{i, v} \in \{0, 1\} \quad \forall v\in V, i\in\lambda\br{v} \label{eqt:integrality_coverage}
\end{align}
Constraints~\eqref{eqt:covering_cost} enforce that the cost of each vertex does not exceed the global cost $M$. Constraints~\eqref{eqt:coverage} ensure that for every edge $uv\in E\br{G}$, there exists a common activated interface. Constraints~\eqref{eqt:costs_coverage} encode the assumption that for every cheap vertex, all of its interfaces are activated, while for every expensive vertex, the cost of activated interfaces is at least $1$. Finally, constraints~\eqref{eqt:integrality_coverage} require the solution to be integral.
Note that the above formulation is not linear due to constraints \eqref{eqt:coverage}. However, by adding additional variables $z_{i, uv}$, we can linearize those by replacing them with the constraints $\sum_{i\in\lambda\br{u}\cap\lambda\br{v}}z_{i, uv} \geq 1$ together with $z_{i, uv} \leq x_{i, u}$ and $z_{i, uv} \leq x_{i, v}$ for every $i\in\lambda\br{u}\cap\lambda\br{v}$ and $uv\in E$. Thus, we can solve the LP relaxation of this formulation in polynomial time.


\subsection{A $k$-approximation algorithm for the Coverage problem}
In this section we show that a simple rounding of the LP solution can be used to obtain a $k$-approximation. Note that for this simple algorithm the preprocessing is not necessary, which makes it easily adaptable for the min-sum variant of the \coverage problem (assuming the ILP is appropriately changed).

\begin{algorithm}[H]
\caption{$k$-approximation algorithm} \label{alg:2}
Solve the LP relaxation and obtain a fractional solution $\{x_{i, v}\}_{i\in\lambda(v), v\in V}$.

Set $\lambda_A(v) \gets \{i \in \lambda(v) \colon x_{i, v} \geq \frac{1}{k}\}$, for every $v \in V$.

\Return{$\{\lambda_A(v)\}_{v\in V}$.}
\end{algorithm}

\begin{theorem}
\Cref{alg:2} is a $k$-approximation algorithm for the \coverage problem.
\end{theorem}

\begin{proof}
Consider an edge $uv\in E$. By averaging, there exists an interface $i\in\lambda\br{u}\cap\lambda\br{v}$ such that $\min\brc{x_{i, u}, x_{i, v}} \geq \frac{1}{k}$. Thus, $i\in\lambda_A\br{u}\cap\lambda_A\br{v}$ so $\lambda_A$ is a covering assignment. Moreover, since $x_{i, v} \geq \frac{1}{k}$ for every $i\in\lambda_A\br{v}$, we have that $c\br{\lambda_A\br{v}} = \sum_{i\in\lambda_A\br{v}}c\br{i, v} \leq k \cdot \sum_{i\in\lambda_A\br{v}}c\br{i, v} \cdot x_{i, v}$. Thus, we have $\COST\br{\lambda_A} \leq k \cdot \max_{v\in V}\sum_{i\in\lambda\br{v}} c\br{i, v} \cdot x_{i, v} \leq k \cdot \OPT$.
\end{proof}

\subsection{An $O\br{\log n}$-approximation algorithm for the \coverage problem}
In order to extract an integral solution with quality independent of $k$, we employ a randomized procedure. For each interface type $i\in\sbr{k}$, we draw a uniformly random threshold $t_i \in \sbr{0, 1}$ and activate the interface $i$ in vertex $v$ if the corresponding variable $x_{i,v}$ scaled up by a factor of $O\br{\log n}$ exceeds $t_i$. The procedure is formalized by the following pseudocode.

\begin{algorithm}[H]
\caption{$O(\log n)$-approximation algorithm for the \coverage problem} \label{alg:1}

Solve the LP relaxation and obtain a fractional solution $\{x_{i, v}\}_{i\in\lambda(v), v\in V}$.

Pick independent uniformly random thresholds $t_i \in \sbr{0, 1}$, for each $i \in \sbr{k}$.

$\lambda_A(v) \gets \{i \in \lambda(v) \colon 2\ln m\cdot x_{i, v} \geq t_i\}$, for every $v \in V$.

\Return{$\{\lambda_A(v)\}_{v\in V}$.}
\end{algorithm}

\begin{theorem}
\Cref{alg:1} is an $O(\log n)$-approximation algorithm for the \coverage problem which succeeds with high probability.
\end{theorem}
\begin{proof}
Firstly, we claim that the algorithm indeed returns a feasible solution:
    \begin{lemma}
With high probability, \Cref{alg:1} returns an assignment $\lambda_A$, such that for every edge $uv \in E$, $\lambda_A(u) \cap \lambda_A(v) \neq \emptyset$.
    \end{lemma}
\begin{proof}
    Consider an edge $uv\in E$. We have:
\begin{align*}
    \Pr\sbr{\lambda_A\br{u}\cap\lambda_A\br{v}=\emptyset} & = \prod_{i\in\lambda\br{u}\cap\lambda\br{v}}\Pr\sbr{t_i > 2\ln m \cdot \min\brc{x_{i, u}, x_{i, v}}}\\
    & = \prod_{i\in\lambda\br{u}\cap\lambda\br{v}}\br{1-2\ln m \cdot\min\brc{x_{i, u}, x_{i, v}}}\\&
    \leq \prod_{i\in\lambda\br{u}\cap\lambda\br{v}}\e{-2\ln m \cdot\min\brc{x_{i, u}, x_{i, v}}} 
    \\& = \e{-2\ln m \cdot\sum_{i\in\lambda\br{u}\cap\lambda\br{v}}\min\brc{x_{i, u}, x_{i, v}}}\\
    & \leq \e{-2\ln m} = \frac{1}{m^2}.
\end{align*}

By applying the union bound over all $m$ edges we get that the algorithm returns a covering assignment with probability at least $1-\frac{m}{m^2} = 1-\frac{1}{m}$.
\end{proof}

Now, we also wish to show that the approximation factor achieved by our algorithm is upper bounded by $O\br{\log n}$. We have the following lemma:
\begin{lemma}
    \Cref{alg:1} returns a solution of cost at most $12\ln n \cdot \OPT$ with high probability.
\end{lemma}
\begin{proof}
Fix an expensive vertex $v\in V$. Otherwise, if $v$ is cheap, then the cost of activating all interfaces at $v$ is at most $1\leq2\cdot\OPT$. We want to show that with high probability, the cost of $v$ is at most $12\ln n$ times the cost of $v$ in the fractional solution. 
By $X_{v, i}$ denote the random variable that is $1$ if $2\ln m\cdot x_{i, v} \geq t_i$ and $0$ otherwise. Then we can write $\COST\br{v, \lambda_A} = \sum_{i\in\lambda\br{v}}c\br{i, v} \cdot X_{v, i}$. We have: $\mathbb{E}\sbr{X_{v, i}} = \Pr\sbr{i \in \lambda_A\br{v}} = 2\ln m\cdot x_{i, v}$. Thus, $\mathbb{E}\sbr{\COST\br{v, \lambda_A}} = 2\ln m \cdot \sum_{i\in\lambda\br{v}}c\br{i, v} \cdot x_{i, v}$. Applying the Chernoff bound with $\delta = 2$ we get 
\begin{align*}
    \Pr\sbr{\COST\br{v, \lambda_A} \geq 3\cdot \mathbb{E}\sbr{\COST\br{v, \lambda_A}}} & 
    \leq \br{\frac{e^2}{27}}^{\mathbb{E}\sbr{\COST\br{v, \lambda_A}}}
    \leq \e{-\mathbb{E}\sbr{\COST\br{v, \lambda_A}}}\leq \frac{1}{m^2}
\end{align*}
where the last inequality is by using the fact that, $\sum_{i\in\lambda\br{v}}c\br{i, v} \cdot x_{i, v} \geq 1$. By applying the union bound over all $n$ vertices we get that with probability at least $1-\frac{n}{m^2}\geq 1-\frac{m+1}{m^2}$ for every expensive $v\in V$, $\COST\br{v, \lambda_A} \leq 6\ln m \cdot \sum_{i\in\lambda\br{v}}c\br{i, v} \cdot x_{i, v}$. Thus, with high probability we have $\COST\br{\lambda_A} \leq 6\ln m \cdot \max_{v\in V}\sum_{i\in\lambda\br{v}}c\br{i, v} \cdot x_{i, v} \leq 6\ln m \cdot \OPT\leq 12\ln n \cdot \OPT$. 
\end{proof}
Combining the above lemmas gives the desired result.
\end{proof}

\section{Approximation algorithm for the \connectivity problem} \label{sec:connectivity}

In this section we move towards the \connectivity problem. We again assume that the instance is preprocessed, so we have that all costs are in $\sbr{0,1}$ and $\OPT \geq 1/2$.

\subsection{ILP formulation of the \connectivity problem}

We use a similar ILP to model the \connectivity problem, however we use additional flow variables to encode the connectivity requirement. The $x_{i, v}$ variables are defined as before. For each edge $e$, we add a variable $y_{e}$ which is $1$ iff the edge $e$ is covered by the solution. For any $S\subseteq V$, let $\delta(S)$ denote the set of edges with exactly one endpoint in $S$. Global connectivity is enforced via cut constraints: for every non-trivial cut $(S, V\setminus S)$, there needs to be at least one covered edge crossing the cut. Let $y\br{\delta\br{S}}=\sum_{e \in \delta(S)} y_e$. We can formulate the \connectivity problem as follows:
\begin{align}
    \text{min} \quad & M \nonumber\\
    \text{s.t.} \quad & \sum_{i\in\lambda\br{v}}c\br{i, v} \cdot x_{i, v} \leq M \quad \forall v\in V \label{eqt:4}\\
    & \sum_{i\in\lambda\br{u}\cap\lambda\br{v}}\min\brc{x_{i, u}, x_{i, v}} \geq y_{uv} \quad \forall uv\in E\\
    &\begin{cases}
            x_{i, v} = 1 \quad \forall i\in\lambda\br{v} & \text{if } \sum_{i\in \lambda\br{v}} c\br{i, v} < 1\\
            \sum_{i\in \lambda\br{v}} c\br{i, v}\cdot x_{i, v} \geq 1 & \text{otherwise}
    \end{cases} \quad \forall v\in V \label{eqt:5}\\
    & y\br{\delta\br{S}}\geq 1 \quad \forall S\subset V\colon S\neq\emptyset,\, S\neq V \label{eqt:6}\\
    & x_{i, v}  \in \{0, 1\} \quad \forall v\in V, i\in\lambda\br{v} \nonumber\\
    & y_{uv} \in \{0, 1\} \quad \forall uv\in E \label{eqt:integrality_connectivity}
\end{align}

Observe that the above ILP has an exponential number of constraints, but we can solve its LP relaxation in polynomial time using the ellipsoid method and a separation oracle for the cut constraints, which can be implemented using a max-flow algorithm.

\subsection{An $O(\log^2 n)$-approximation algorithm for the \connectivity problem}
Obtaining an integral solution for the \connectivity problem requires a more careful rounding procedure than for the \coverage problem. We wish to ensure that at least one connected spanning subgraph is covered by the assignment returned by the algorithm. To do so, we start by randomly sampling edges according to the $y$-variables of the fractional solution, in a way so that with high probability, the sampled subgraph $H$ is connected. To achieve this goal, we set the probability of choosing $e$ to be $p_e=\min\{1, 5\ln m \cdot y_e\}$. Then, the goal of the algorithm becomes covering $H$. To do so, we repeat the following process: we draw uniformly random thresholds $t_i$ for each type of interface $i\in \sbr{k}$, and activate $i$ at any vertex $v$ for which $5\ln m \cdot x_{i,v}\geq t_i$. Intuitively, each such iteration costs $O\br{\log n\cdot \OPT}$ in expectation and covers $\Omega\br{1}$ fraction of edges of $H$. After repeating this for a fixed number of $T = O\br{\log n}$ rounds, it can be shown that the algorithm covers all edges of $H$ almost surely. The formal description of the algorithm is given by the following pseudocode.

\begin{algorithm}[H]
\caption{$O\br{\log^2n}$-approximation algorithm for the \connectivity problem} \label{alg:3} 

Solve the LP relaxation and obtain a fractional solution $\{x_{i, v}\}_{i\in\lambda(v), v\in V}, \{y_{uv}\}_{uv\in E}$.

For every edge $e\in E$, $p_e \gets \min\{1, 5\ln m \cdot y_e\}$. 

With probability $p_e$ set $Y_e \gets 1$, otherwise set $Y_e \gets 0$. 

Let $H$ be the subgraph of $G$ induced by the edges $e$ such that $Y_e = 1$.

$T \gets \frac{2\ln m}{1-1/e}$.

$\lambda_A\gets \{\emptyset\}_{v\in V}$.

\For{$j \in \sbr{T}$}{
Pick independent uniformly random thresholds $t_i \in (0, 1)$, for each $i \in \sbr{k}$.

$\lambda_A^j(v) \gets \{i \in \lambda(v) \colon 5\ln m \cdot x_{i, v} \geq t_i\}$, for every $v \in V$.

$\lambda_A \gets \lambda_A^j \cup \lambda_A$.
}
\Return{$\lambda_A$}
\end{algorithm}

\begin{theorem}
    \Cref{alg:3} is an $O\br{\log^2 n}$-approximation algorithm for the \connectivity problem which succeeds with high probability.
\end{theorem}
\begin{proof}
    In order to show that the algorithm returns a connecting assignment, we first show that with high probability, the subgraph $H$ is connected. It should be remarked that constructing $H$ explicitly is not necessary for the algorithm to work. Nonetheless, we do so for the sake of convenience of our analysis. We start with the following lemma:
    \begin{lemma}
        With high probability, $H$ is connected.
    \end{lemma}
    \begin{proof}
        Fix any $S\subset V$ with $S\neq\emptyset$ and $S\neq V$. By definition of the LP, there exists a natural number $\ell\in\mathbb{N}$, such that $\ell\leq y\br{\delta_H\br{S}}<\ell+1$. Let $\delta_H\br{S}=E\br{H}\cap \delta\br{S}$.
        We have:
        \begin{align*}
            \Pr\sbr{\spr{\delta_H\br{S}}=0} & = \prod_{e\in\delta\br{S}}\Pr\sbr{Y_e=0} = \prod_{e\in\delta\br{S}}(1-p_e) \\&\leq \prod_{e\in\delta\br{S}}\e{-p_e} \leq \e{-\sum_{e\in\delta(S)}\min\brc{1,\, 5\ln m\cdot y_e}}
            \\& \leq \e{-5\ell\ln m} = \frac{1}{m^{5\ell}}.
        \end{align*}

        A cut in $G$ is called an $\alpha$-minimum cut if its value is less than $\alpha$ times the value of a minimum cut in $G$. 
        We now make use of the following Karger's lemma \cite{MinimumCutsInNearLinearTime, RandomSamplingInCutFlowAndNetworkDesignProblems}:
        \begin{lemma}[Karger's lemma]
            The number of $\alpha$-minimum cuts in a graph is at most $n^{2\alpha}$.
        \end{lemma}
        Observe that by construction of the LP, every non-trivial fractional cut has value at least $1$, so there are at most $n^{2\ell+2}$ such cuts that have fractional value less than $\ell+1$. Applying the union bound over all of them we obtain:
\begin{align*}
    \Pr\sbr{\exists\, S\subset V,\; S\neq\emptyset,\, S\neq V\colon \spr{\delta_H\br{S}}=0}
    &\leq \sum_{\ell=1}^{\infty}\sum_{\substack{S\subset V,\,S\neq\emptyset,\,S\neq V,\\ y\br{\delta_H\br{S}}<\ell+1}}\Pr\sbr{\spr{\delta_H\br{S}}=0} \\
    &\leq \sum_{\ell=1}^{\infty}n^{2\ell+2}\cdot \frac{1}{m^{5\ell}} \leq \sum_{\ell=1}^{\infty}\frac{1}{m^{\ell}} = \frac{1}{m-1}
\end{align*}
where the last inequality is due to the fact, that assuming that $G$ is not a tree, we get that $m\geq n$. Otherwise, \coverage and \connectivity are equivalent and one can use the $O\br{\log n}$-approximation algorithm instead.

Thus, with probability at least $1-\frac{1}{m-1}$, every non-trivial cut $S$ of $G$, has at least one edge of $H$ crossing it, which means that $H$ spans all vertices of $G$.
\end{proof}

Having shown that the subgraph $H$ is connected, we now wish to show that the assignment $\lambda_A$ returned by the algorithm covers all of the edges of $H$ almost surely. 
Let $j$ be an index of some iteration of the for loop. Let $H_j$ be the (possibly empty) subgraph of $H$ that is not covered by $\lambda_A$ at the beginning of the $j$-th iteration of the for loop. We wish to show, that in expectation, the solution $\lambda_A$ produced in the $j$-th iteration of the for loop, $\lambda_A^j$, covers a significant fraction of the edges in $H_j$. We have the following lemma:
\begin{lemma}
    We have
    $\mathbb{E}\sbr{\spr{\cov\br{\lambda_A^j, H_j}}} \geq \br{1-1/e} \cdot \mathbb{E}\sbr{\spr{E\br{H_j}}}$.
\end{lemma}
\begin{proof}
    Consider any edge $uv\in E\br{G}$. Observe that $\mathbb{E}\sbr{Y_{uv}}\leq 5\ln m \cdot y_{uv}$. By definition of the LP we get that \[
        \sum_{i\in\lambda\br{u}\cap\lambda\br{v}}\min\brc{x_{i, u}, x_{i, v}} \geq y_{uv} \geq \frac{\mathbb{E}\sbr{Y_{uv}}}{5\ln m}.\]
    Therefore, the probability that $uv$ is uncovered by $\lambda_A^j$ is at most:
    \begin{align*}
        \Pr\sbr{\lambda_A\br{u}^j\cap\lambda_A\br{v}^j=\emptyset} & = \prod_{i\in\lambda\br{u}\cap\lambda\br{v}}\Pr\sbr{t_i > 5\ln m\cdot\min\brc{x_{i, u}, x_{i, v}}} \\ 
        & = \prod_{i\in\lambda\br{u}\cap\lambda\br{v}}\br{1-5\ln m\cdot\min\brc{x_{i, u}, x_{i, v}}} \\& \leq \prod_{i\in\lambda\br{u}\cap\lambda\br{v}}\e{-5\ln m\cdot\min\brc{x_{i, u}, x_{i, v}}}\\ 
        & = \e{-5\ln m\cdot\sum_{i\in\lambda\br{u}\cap\lambda\br{v}}\min\brc{x_{i, u}, x_{i, v}}} \\& \leq \e{-5\ln m\cdot\frac{\mathbb{E}\sbr{Y_{uv}}}{5\ln m}}=\e{-\mathbb{E}\sbr{Y_{uv}}}.
    \end{align*}
    
    Thus, we have $\Pr\sbr{\lambda_A\br{u}^j\cap\lambda_A\br{v}^j\neq\emptyset}\geq 1-\e{-\mathbb{E}\sbr{Y_{uv}} }$. By summing up over all edges in $H_j$ we lower bound the expected number of edges covered by $\lambda_A^j$ as follows:
    \begin{align*}
        \mathbb{E}\sbr{\spr{\cov\br{\lambda_A^j, H_j}}} & = \sum_{uv\in E\br{H_j}}\Pr\sbr{\lambda_A\br{u}^j\cap\lambda_A\br{v}^j\neq\emptyset}  
        \geq \sum_{uv\in E\br{H_j}}\br{1-\e{-\mathbb{E}\sbr{Y_{uv}} }}
        \\& \geq \br{1-1/e}\cdot \sum_{uv\in E\br{H_j}}{\mathbb{E}\sbr{Y_{uv}}} \geq \br{1-1/e} \cdot \mathbb{E}\sbr{\spr{E\br{H_j}}}
    \end{align*}
    where the second inequality is due to the fact that $1-e^{-x}\geq (1-1/e)\cdot x$, for $x\in\sbr{0,1}$ and the last inequality is obtained by linearity of expectation.
\end{proof}

Let $I$ be the random variable denoting the number of iterations needed to cover all edges of $H$, i.e., the smallest $j$ such that $\cov\br{\lambda_A, H} = E\br{H}$ after $j$ iterations (assume that the number of iterations might eventually get larger then $T$). We show that with high probability $I \leq T$, so the algorithm returns a feasible solution.
\begin{lemma}
    \label{lem:expected-iterations}
    $\Pr\sbr{I > T} = O\br{\frac{1}{m}}$.
\end{lemma}
\begin{proof}
    Let $\alpha = 1-1/e$. By definition, $\spr{E\br{H_{j+1}}}=\spr{E\br{H_{j}}}-\spr{\cov\br{\lambda_A^j, H_j}}$. By linearity of expectation and the previous lemma, $\mathbb{E}\sbr{\spr{E\br{H_{j+1}}}} \leq \br{1-\alpha}\cdot \mathbb{E}\sbr{\spr{E\br{H_{j}}}}$. Unfolding the recurrence gives $\mathbb{E}\sbr{\spr{E\br{H_{j}}}} \leq m e^{-\alpha(j-1)}$. By Markov's inequality, $\Pr\sbr{\spr{E\br{H_{j}}}\geq 1} \leq m e^{-\alpha(j-1)}$. Since $T = \frac{2\ln m}{1-1/e} = \frac{2\ln m}{\alpha}$, setting $j = T$ gives:
    $$\Pr\sbr{I > T} = \Pr\sbr{\spr{E\br{H_{T}}}\geq 1} \leq m e^{-\alpha(T-1)} \leq m e^{-2\ln m+\alpha} = O\br{\frac{1}{m}}$$
    since $e^{\alpha}=O\br{1}$.
\end{proof}

It remains to show that the cost of the solution returned by the algorithm is at most $O\br{\log^2 n}\cdot \OPT$. We have with the following lemma:
\begin{lemma}
    With high probability, for every vertex $v\in V\br{G}$ and every iteration $j \in \sbr{T}$, $\COST\br{v, \lambda_A^j} \leq 15\ln m \cdot \OPT$.
\end{lemma}
\begin{proof}
Fix a vertex $v\in V$ and an iteration $j\in\sbr{T}$. If $v$ is cheap, then its overall cost is at most $1\leq 2\cdot\OPT$ regardless of activated interfaces. Thus, assume $v$ to be expensive, so $\sum_{i\in\lambda\br{v}}c\br{i,v} \cdot x_{i,v} \geq 1$.
By $X_{v,i}^j$ denote the indicator that is true iff $i\in \lambda_A^j$. Then, $\COST\br{v, \lambda_A^j} = \sum_{i\in\lambda\br{v}}c\br{i,v} \cdot X_{v,i}^j $. We have $\mathbb{E}\sbr{X_{v,i}^j} = 5\ln m\cdot x_{i,v}$, hence $\mathbb{E}\sbr{\COST\br{v,\lambda_A^j}} = 5\ln m\cdot \sum_{i\in\lambda\br{v}}c\br{i,v} \cdot x_{i,v}$. Applying the Chernoff bound with $\delta=2$ gives:
\begin{align*}
    \Pr\sbr{\COST\br{v, \lambda_A^j} \geq 3\cdot\mathbb{E}\sbr{\COST\br{v,\lambda_A^j}}}
    \leq \br{\frac{e^2}{27}}^{\mathbb{E}\sbr{\COST\br{v,\lambda_A^j}}}
    \leq e^{-\mathbb{E}\sbr{\COST\br{v,\lambda_A^j}}}
    \leq \frac{1}{m^5}
\end{align*}
where the last inequality uses $\mathbb{E}\sbr{\COST\br{v,\lambda_A^j}} = 5\ln m\cdot\sum_{i\in\lambda\br{v}}c\br{i,v} \cdot x_{i,v} \geq 5\ln m$. Hence $\COST\br{v,\lambda_A^j} \leq 15\ln m\cdot\sum_{i\in\lambda\br{v}} c\br{i,v} \cdot x_{i,v} \leq 15\ln m\cdot\OPT$ with probability at least $1-\frac{1}{m^5}$.

Applying the union bound over all $n \leq m + 1$ vertices and all $T \leq \frac{2\ln m}{1-1/e}$ iterations, the above holds for all $v\in V$ and $j\in\sbr{T}$ with probability at least $1 - \frac{T\cdot \br{m+1}}{m^5} = 1 - O\br{\frac{1}{m^3}}$.
\end{proof}
Consequently, summing up over all $T$ iterations, we get that with high probability:
$$\COST\br{\lambda_A} \leq \sum_{j=1}^{T}\max_{v\in V\br{G}}\brc{\COST\br{v, \lambda_A^j}} \leq T\cdot 15\ln m\cdot\OPT = \frac{30\ln^2 m}{1-1/e}\cdot\OPT\leq \frac{120\ln^2 n}{1-1/e}\cdot\OPT.$$
Combining this with the above lemmas gives the desired result.
\end{proof}

%% file: sections/conclusion.tex
\section{Conclusion} \label{sec:conclusion}

In this work, we have provided the first polylogarithmic approximation algorithms for both \connectivity and \coverage problems, which are fundamental tasks in covering multi-interface networks. Specifically, we have provided an $O\br{\log^2 n}$-approximation for the former as well as an $O\br{\log n}$-approximation for the latter, which is almost tight since \coverage is NP-hard to approximate within a factor of $o\br{\log \Delta}$, even when the input graph is a star and all costs are uniform~\cite{MinimizeTheMaximumDutyInMultiInterfaceNetworks}. This inapproximability carries on also for the \connectivity problem, since on stars, \connectivity is equivalent to \coverage. The remaining task is to close the gap between the obtained $O\br{\log^2 n}$ and $\Omega\br{\log \Delta}$ lower bound.


%% file: sections/appendix.tex
\section{Weighted Chernoff Bound Derivation} \label{sec:appendix}

\begin{theorem}[(Weighted) Chernoff bound]
    Let $X_1, X_2, \dots, X_n$ be independent random variables taking values in $\sbr{0,1}$, let $w_1, w_2, \dots, w_n$ be non-negative weights, such that for every $i\in \sbr{n}$ we have that $w_i\in \sbr{0,1}$. Let $X=\sum_{i=1}^{n} w_i\cdot X_i$ and let $\mu = \E{X}$. Then for every $\delta_1 > 0$ and $\delta_2 \in (0,1)$ we have that:
    \[
        \Pr\sbr{X \geq (1+\delta_1)\cdot \mu} \leq \br{\frac{e^{\delta_1}}{(1+\delta_1)^{1+\delta_1}}}^\mu
        \qquad
        \Pr\sbr{X \leq (1-\delta_2)\cdot \mu} \leq \br{\frac{e^{-\delta_2}}{(1-\delta_2)^{1-\delta_2}}}^\mu.
    \]
\end{theorem}

\begin{proof}
We first prove the upper-tail bound for $\delta_1>0$. For any $\lambda>0$, Markov's inequality gives:
\[
\Pr\sbr{X \ge \br{1+\delta_1}\cdot\mu}
= \Pr\sbr{e^{\lambda X} \ge e^{\lambda\cdot\br{1+\delta_1}\cdot\mu}} \le \frac{\E{e^{\lambda X}}}{e^{\lambda\cdot\br{1+\delta_1}\cdot\mu}}.
\]
By independence, we have $
\E{e^{\lambda \cdot X}}
= \prod_{i=1}^n \E{e^{\lambda\cdot w_i\cdot X_i}}$.
Fix $i\in\sbr{n}$. Since $w_i\cdot X_i\in\sbr{0,1}$, by convexity of $y\mapsto e^{\lambda \cdot y}$ on $\sbr{0,1}$, for every $y\in\sbr{0,1}$,
$e^{\lambda \cdot y}
\le \br{1-y}\cdot e^0 + y\cdot e^{\lambda}
= 1 + \br{e^{\lambda}-1}\cdot y$.
Substituting $y=w_i\cdot X_i$, taking expectation, and using $1+t\le e^t$, we obtain:
\[
\E{e^{\lambda\cdot w_i\cdot X_i}}
\le 1 + \br{e^{\lambda}-1}\cdot w_i\cdot\E{X_i} \le \exp\br{\br{e^{\lambda}-1}\cdot w_i\cdot\E{X_i}}.
\]
Hence:
\[
\E{e^{\lambda X}}
\le \exp\br{\br{e^{\lambda}-1}\cdot\sum_{i=1}^n w_i\cdot\E{X_i}} = \exp\br{\br{e^{\lambda}-1}\cdot\mu}.
\]
Therefore,
$\Pr\sbr{X \ge \br{1+\delta_1}\cdot\mu}
\le \exp\br{\mu\cdot\br{e^{\lambda}-1-\lambda\cdot\br{1+\delta_1}}}$.
The function $f\br{\lambda}=e^{\lambda}-1-\lambda\cdot \br{1+\delta_1}$ is minimized at $\lambda=\ln\br{1+\delta_1}$, and so:
\[
\Pr\sbr{X \ge \br{1+\delta_1}\cdot\mu}
\le \exp\br{-\mu\cdot\br{\br{1+\delta_1}\cdot\ln\br{1+\delta_1}-\delta_1}} = \br{\frac{e^{\delta_1}}{\br{1+\delta_1}^{1+\delta_1}}}^{\mu}.
\]

Now let $\delta_2\in(0,1)$. For any $t>0$, again by Markov's inequality:
\[
\Pr\sbr{X \le \br{1-\delta_2}\cdot\mu}
= \Pr\sbr{e^{-tX} \ge e^{-t\cdot\br{1-\delta_2}\cdot\mu}} \le \frac{\E{e^{-t\cdot X}}}{e^{-t\cdot\br{1-\delta_2}\cdot\mu}}.
\]
Applying the same mgf estimate with $\lambda=-t<0$ yields $
\E{e^{-t\cdot X}}
\le \exp\br{\br{e^{-t}-1}\cdot\mu}$,
so:
\[
\Pr\sbr{X \le \br{1-\delta_2}\cdot\mu}
\le \exp\br{\mu\cdot\br{e^{-t}-1+t\cdot\br{1-\delta_2}}}.
\]
The right-hand side is minimized at $t=-\ln\br{1-\delta_2}$, which gives:
\[
\Pr\sbr{X \le \br{1-\delta_2}\cdot\mu}
\le \exp\br{-\mu\cdot\br{\delta_2+\br{1-\delta_2}\cdot\ln\br{1-\delta_2}}} = \br{\frac{e^{-\delta_2}}{\br{1-\delta_2}^{1-\delta_2}}}^{\mu}.
\]
\end{proof}

\section{Preprocessing Details} \label{sec:appendix-preprocessing}

In this section, we provide the full preprocessing routine used for both the \coverage and \connectivity problems.

\subparagraph*{Scaling of the costs}
Firstly, we produce a polynomial number of instances, in such a way that for at least one of them we have a good lower bound on $\OPT$. For convenience, we allow costs to be rational numbers and scale them appropriately so that all of them lie in the interval $\sbr{0,1}$. To do so, by testing consecutive powers of $2$, we guess the smallest number $b$, such that the costliest activated interface among all vertices in the optimal solution has cost at most $2^b$. Let $C=\cl{\log\br{\max_{v\in V\br{G}, i\in \lambda\br{v}} \brc{c(i,v)}}}$. The number of the guesses is $C+1$, which is polynomial in the input size. If either some edge cannot be covered by interfaces of cost at most $2^b$, or there is no vertex $v$ and interface $i\in \lambda\br{v}$ such that $c\br{i,v}\in\sbr{1/2,1}$, then we discard such a guess. Otherwise, for each vertex $v\in V$, we temporarily discard from $\lambda\br{v}$ all interfaces with cost larger than $2^b$. Let $C_b=\max_{v\in V\br{G}, i\in \lambda\br{v}} \brc{c(i,v)}$ in this newly created instance. We normalize the costs of remaining interfaces by dividing them by $C_b$, so that the largest cost is exactly $1$. For the right guess of $b$, this gives $\OPT\geq C_b\geq 1/2$.

\subparagraph*{Cheap and expensive vertices}
Having established the scaling of costs, we now consider two types of vertices with respect to the new costs. For a vertex $v\in V$, if $\sum_{i\in \lambda\br{v}} c\br{i,v}\leq 1$, then we can safely assume that $v$ activates all interfaces in $\lambda\br{v}$, since this costs at most $1$ and is thus upper bounded by $2\cdot\OPT$. We call such vertices \emph{cheap}. The remaining vertices are \emph{expensive}. For every expensive vertex $v$, we can assume that for any output of the algorithm, $\COST\br{\lambda_A,v}\geq 1$, since enforcing this changes the cost by a factor of at most $2$ and does not affect asymptotic guarantees.

\subparagraph*{Finding the solution}
Since our algorithms are randomized, we also need to control the probability of failure. We ensure that, with high probability, for all values of $b$ we find a good solution. To do so, for each guess of $b$, we repeat the randomized procedure $K=\cl{\log_m\br{C}+1}$ times (which is polynomial). The algorithm is as follows.

\begin{algorithm}[H]
\caption{Preprocessing algorithm} \label{alg:preprocessing}
$\lambda_A(v) \gets \lambda(v)$ for every $v\in V$.

\For{$b\in \brc{0, \dots, C}$}
{
    For the remainder of this iteration, discard interfaces with cost larger than $2^b$ and divide all remaining costs by $\max_{v\in V\br{G}, i\in \lambda\br{v}} \brc{c(i,v)}$.

    \If{there is no solution for this new instance or there is no vertex $v$ and interface $i\in \lambda\br{v}$ such that $c\br{i, v}\in \sbr{1/2, 1}$}
    {
        \textbf{continue}
    }

    \For{$j\in \sbr{K}$}
    {
        $\lambda_A' \gets$ the output of the randomized procedure for the scaled-cost instance.

        \If{$\COST\br{\lambda_A'} < \COST\br{\lambda_A}$ and $\lambda_A'$ is a covering/connecting assignment}
        {
            $\lambda_A \gets \lambda_A'$.
        }
    }
}

\Return{$\{\lambda_A(v)\}_{v\in V}$.}
\end{algorithm}

Assuming that the failure probability of the randomized procedure is at most $O\br{1/m}$, for each fixed $b$ the probability of failure after $K$ repetitions is at most $O\br{1/m^K}=O\br{1/(mC)}$. Since there are $C$ guesses for $b$, a union bound implies that the probability that there exists some value of $b$ for which the algorithm fails is at most $O\br{1/m}$. In particular, this also holds for the correct guess of $b$.